\documentclass[12pt, titlepage]{article}
\usepackage{eurosym}
\usepackage{amsmath,amssymb,amsthm,amsfonts}
\usepackage{epstopdf}
\usepackage{subcaption}
\usepackage{appendix}
\usepackage{amsmath}
\usepackage{amsfonts}
\usepackage{amssymb}
\usepackage{multirow}
\usepackage{rotating}
\usepackage{color}
\usepackage{lscape}
\usepackage{verbatim}
\usepackage{graphicx}

\numberwithin{equation}{section} 

\newtheorem{theorem}{Theorem}[section]
\newtheorem{corollary}{Corollary}[section]

\newtheorem{proposition}{Proposition}[section]
\newtheorem{definition}{Definition}[section]
\newtheorem{remark}{Remark}

\def\XX{\boldsymbol{X}}

\def\rr{\boldsymbol{r}}

\def\AA{\boldsymbol{A}}

\def\RRR{\boldsymbol{R}}

\def\TT{\boldsymbol{\Theta}}
\def\II{\boldsymbol{I}}

\def\EEE{\mathcal{E}_d}

\def\SSS{\mathcal{S}_d}

\def\YY{\boldsymbol{Y}}

\def\pp{\boldsymbol{p}}

\def\RR{\mathbb{R}}

\def\zz{\boldsymbol{z}}

\def\ff{\boldsymbol{f}}

\def\Rjj{R_{(j_1,j_2)}}
\def\rjj{\rr_{(j_1,j_2)}}
\def\design{\mathcal{X}_d}
\def\DDD{\mathcal{D}_d}
\def\VaR{\text{VaR}_{\alpha}}
\def\ES{\text{ES}_{\alpha}}

\DeclareMathOperator\expval{E}

\DeclareMathOperator{\rank}{rank}

\begin{document}

\author{ R.FONTANA \\ \textit{\ 
Department of Mathematical Sciences G. Lagrange,} \\ {Politecnico di
Torino.} \\ E. LUCIANO \thanks{%
Elisa Luciano gratefully acknowledges financial support from the Italian
Ministry of Education, University and Research (MIUR), "Dipartimenti di
Eccellenza" grant 2018-2022.}\\ \textit{\ 
ESOMAS\ Department and Collegio Carlo Alberto, }Universit\`{a} di Torino%
\\ P. SEMERARO\thanks{%
Roberto Fontana and Patrizia Semeraro gratefully acknowledge financial support from the Italian
Ministry of Education, University and Research (MIUR), "Dipartimenti di
Eccellenza" grant 2018-2022.}  \\ \textit{\ 
Department of Mathematical Sciences G. Lagrange,} \\ {Politecnico di
Torino.}}
\title{Model Risk in Credit Risk}
\maketitle

\begin{abstract}

The issue of model risk in default modeling has been known since inception
of the Academic literature in the field.  However, a rigorous treatment requires a description of all
the possible models, and a measure of the distance between a single model
and the alternatives, consistent with the applications. This is the purpose
of the current paper. We first analytically describe all possible joint models for
default, in the class of finite sequences of exchangeable Bernoulli random variables. We then
measure how the model risk of choosing or calibrating one of them affects
the portfolio loss from default, using two popular and economically
sensible metrics, Value-at-Risk (VaR) and Expected Shortfall (ES).

\emph{keywords}: Exchangeable Bernoulli distribution; risk measures; model risk.
\end{abstract}


\section{Introduction}

Models for default risk are prone to so-called model risk, in two senses:
adopting the wrong model for the occurrence of default and calibrating or
estimating a given model in a wrong way. The occurrence of model risk in the
first sense is inherent in default, because of the difficulty of describing
the causes of default or even of enumerating the determinants. Even the
occurrence of calibration or estimation risk is overwhelming, because of the
scarcity of observations, especially when looking at the joint default of
specific obligors or particular categories of obligors, and lack of data to
estimate parameters such as the correlation of defaults. The issue of model
risk is indeed particularly strong in joint defaults, because on top of the
model risk for marginal defaults there is model risk also in their joint
distribution. We focus on joint modeling.

The issue of model risk in default modeling has been known since inception
of the Academic literature in the field. Professionals are well aware of its
importance too. However, a rigorous treatment requires a description of all
the possible models and a measure of the distance between a single model
and the alternatives, consistent with the applications. This is the purpose
of the current paper. We first describe all possible joint models for
default, in the class of exchangeable Bernoulli random variables. We then
measure how the model risk of choosing or calibrating one of them affects
the portfolio loss from default, using two popular and economically
sensible metrics, Value-at-Risk (VaR) and Expected Shortfall (ES).

Univariate models of default belong to two families: structural and
reduced-form models. The structural models, initiated by \cite{merton1974pricing},
reconduct default to the fact that the so-called asset value of a firm goes
below a given monetary threshold. Reduced-form models, whose seminal work is
due to \cite{jarrow1992pricing}, estimate from interest rates on
defaultable debt the intensity of default, which is then interpreted as a
fixed parameter or a stochastic process itself. For a survey of the
approaches see for instance \cite{bielecki2009credit}.
Multivariate models either make use of a copula to aggregate univariate
default probabilities (see for instance \cite{cherubini2004copula}, or \cite{duffie2003credit}, or use a Bernoulli mixture model
(see chapter 8 in \cite{mcneil2005quantitative}).

The difficulties in choosing the right model for univariate modeling and
calibrating it have been shown to be considerable. For structural models,
the asset value is unobservable. For reduced-form models, rates of return on
bonds are thought to include also a liquidity spread, which is difficult to
separate from the default spread.

The difficulties in choosing or calibrating a multivariate model are even
bigger (see the early recognition in \cite{embrechtscorrelation}). Structural models can be calibrated, provided the correlation
matrix of asset values can be. Multivariate reduced-form models are usually
calibrated using the corresponding structural dependence (see chapter 10 in
\cite{duffie2003credit}).

The previous literature which assesses model risk in joint default usually
takes as given the marginal probabilities of default, as we do:\ marginal
default indicators are Bernoulli variables. It tries to explore the range of
joint default probabilities, or the possible distribution of the loss from
credit risk, which is the weighted sum of the marginal Bernoulli variables,
where the weights are the exposures of the creditor towards different
obligors. To do that, the literature uses different copulas (see\cite{embrechts2003using}). Here we use the fact that all joint distributions
or distributions of sums are generated starting from a finite number of
so-called ray densities. Differently from copulas, all the rays can be
found, either numerically or analytically.

\cite{fontana2018representation} developed a simple method to represent all
the Bernoulli variables with some specified moments, as a convex hull of
densities belonging to the same class, the ray densities. They provide an
algorithm to find the extreme rays of a given class without restrictions
either on the number of variables or on the specified moments. The only
drawback of the method is the amount of computational effort required for
the numerical solution. The main contribution of the current paper consists
in finding analytically the convex hull generators for the class of exchangeable Bernoulli variables with given mean and for the class  of exchangeable Bernoulli variables with given mean and correlation. The analytical solution allows us to work  in any dimension.

Once the multivariate Bernoulli variables represent the default indicators
of a portfolio of obligors, the  ray densities, that we can find analytically,  allow
us to describe all the joint distributions of defaults, even for large
portfolios, and/or the possible distributions of the loss. There is a third
mathematical contribution that helps in doing that:  we show that the VaR  bounds are reached on ray densities and we find an analytical expression for them. We also explicitly found bounds for the ES.
We then measure the consequence of using a specific model (which might be
"wrong" one) or calibrating it in the "wrong" way looking at the range of
the possible VaR and ES.

So, the paper is novel both for the Mathematical contribution, namely the
analytical description of the ray densities in high dimensions, and for the
Mathematical Finance one, namely measurement of model risk using all
possible multivariate distributions,  obtained as linear convex combinations of generators that can be analytically found. This analytical solution allows us to find analogical bounds to measure model risk.

The paper unfolds as follows: Section \ref{not} introduces the mathematical framework. Section \ref{EG} introduces the notion and properties of rays
for exchangeable Bernoulli variables. Section \ref{FM} introduce the risk measures and provide analytical bounds for exchangeable Bernoulli variables.  Model risk is discussed in Section %
\ref{MR}. Section \ref{ex} provides calibrated examples. Section \ref{fine}
concludes.
\section{\protect\bigskip Default indicators: mathematical background\label{not}}

We consider a credit portfolio $P$ with $d$ obligors.

Some notation is needed.
Let the random variable $\boldsymbol{X}=(X_1, \ldots, X_{d})$ be the default
indicators for the portfolio $P$ and let us assume that the indicator $%
\boldsymbol{X}$ is exchangeable, i.e. $\boldsymbol{X}\in \mathcal{E}_d$, where
$\mathcal{E}_d$ is the class of $d$-dimensional exchangeable Bernoulli
distributions. Let $\mathcal{E}_d(p)$ be the class of exchangeable Bernoulli
distributions with the same Bernoulli marginal distributions $B(p)$, where $%
p $ is the marginal default probability of each obligor. If $\boldsymbol{X}%
=(X_1, \dots, X_d)$ is a random vector with joint distribution in $\mathcal{E%
}(p)$, we denote

\begin{itemize}
\item its cumulative distribution function by $F_{p}$ and its probability  mass function (pmf)
by $f_{p}$;
\item the column vector which contains the values of  $f_{p}$
over $\design:=\{0, 1\}^d$, by $ (f_{p}(\boldsymbol{x}):\boldsymbol{x}\in\design)$
respectively; we make the non-restrictive hypothesis that the set $\design$ of $2^d$ binary vectors is ordered according to the
reverse-lexicographical criterion. For example $\mathcal{X}_2=\{00, 10, 01,
11\}$ and $\mathcal{X}_3=\{000, 100, 010, 110, 001, 101, 011, 111\}$;

\item we denote by $\mathcal{P}_d$ the set of permutations on $\{1,\ldots, d\}$;



\end{itemize}

Recall that the expected value of $X_{i}$ is $p$, $\expval[X_i]=p,\,\,\, i=1,\ldots, d$. We denote $q=1-p$.
%
We assume that vectors are column vectors. 

\subsection{Exchangeable Bernoulli variables}

Let us consider a pmf $f_p$ of a $d$-dimensional Bernoulli distribution with mean $p$.
Since $f_{p}(\boldsymbol{x})=f_{p}(\sigma (\boldsymbol{x}))$ for any $%
\sigma \in \mathcal{P}_d$, any mass function $f_p$  in $\EEE(p)$ is given by $f_{i}:=f_{p}(\boldsymbol{x})$ if $%
\boldsymbol{x}=(x_{1},\ldots ,x_{d})\in \design$ and $%
\#\{x_{j}:x_{j}=1\}=i$.
Therefore  we identify a mass function $f_p$ in $\EEE(p)$ with the corresponding vector $\ff_p:=(f_0, \ldots, f_d)$. Furthermore,  the moments depend only on their order, we therefore use $%
\mu_{{\alpha}}$ to denote a moment of order $\alpha=\text{ord}(%
\boldsymbol{\alpha})=\sum_{i=1}^d\alpha_i$, where $\boldsymbol{\alpha}\in \design$.

We also observe that the
correlation $\rho$ between two Bernoulli variables $X_i \sim B(p)$ and $X_j
\sim B(p)$ is related to the second-order moment $\mu_{2}=
\expval[X_i
X_j]$ as follows
\begin{equation}  \label{eq:rho_e12}
\mu_2=\rho pq+p^2.
\end{equation}

\subsection{Joint defaults, loss distribution and risk measures}

To model the loss of a credit risk portfolio $P$ of $d$ obligors we consider
the sum of the individual losses
\begin{equation*}
L=\sum_{i=1}^{d}w_{i}X_{i},
\end{equation*}
where $w_{i}\in (0,1]$ and $\sum_{1=1}^{d}w_{i}=1$. In this paper we consider the case $w_{i}=\frac{1}{d},\,i\in
\{1,\ldots ,d\}$. The extension to unequal weights can be done numerically. For
equal weights, $L=\frac{S_{d}}{d}$, where
\begin{equation*}
S_{d}=\sum_{i=1}^{d}X_{i}
\end{equation*}%
represents the number of defaults. Therefore, the distribution of $S_{d}$
represents the distribution of the loss.  Since the vector of default indicators $\XX$ is assumed to be exchangeable, there is a one-to-one  correspondence between the distribution of the number of defaults and  the joint distribution of $\XX$.
In fact, as said in the preliminaries,
 since $f_{p}(\boldsymbol{x})=f_{p}(\sigma (\boldsymbol{x}))$ for any $%
\sigma \in \mathcal{P}_d$, any mass function $f_p$  in $\EEE(p)$ is given by $f_{i}:=f_{p}(\boldsymbol{x})$ if $%
\boldsymbol{x}=(x_{1},\ldots ,x_{d})\in \mathcal{D}_{d}$ and $%
\#\{x_{j}:x_{j}=1\}=i$. We can define
a one-to-one correspondence between $\EEE(p)$ and the class of the distributions on the number of defaults.

 Let $\SSS(p)$ be the class of distributions $p_S$ on $\{0,\ldots, d\}$ such that $S_d=\sum_{i=0}^dX_i$ with  $\XX\in \EEE(p)$. Let $p_S(j)=p_j=P(S_d=j)$ and $\pp_S=(p_0,\ldots, p_d)$.

The map:
\begin{equation}\label{map}
\begin{split}
E: \EEE(p)&\rightarrow \mathcal{S}_d(p)\\
f_{j} &\rightarrow p_j={\binom{d}{j}}f_j.
\end{split}
\end{equation}
is a one-to-one correspondence between $ \EEE(p)$ and $\SSS(p)$.
%
Therefore we have
\begin{equation}\label{map}
\begin{split}
\EEE(p)&\leftrightarrow \mathcal{S}_d(p)
\end{split}
\end{equation}
 We now prove that the class of distributions $\SSS(p)$ coincides with the entire class of discrete distributions  with mean $dp$, say $\DDD(dp)$. This fact is useful to simplify the search of the generators of $\EEE(p)$. The class $\DDD(dp)$ is not of special interest in this context, but it is introduced for technical reasons.
\begin{proposition}\label{d-eq-s}
It holds $\SSS(p)=\DDD(dp)$.
\end{proposition}
\begin{proof}
\begin{description}
\item{1)} $\SSS(p)\subseteq \DDD(dp)$. This is trivial.

\item{2)} $\DDD(dp)\subseteq \SSS(p)$. Let $\{p_0,\ldots, p_d\}\in \DDD(dp)$. Let us define $f_i=\frac{p_i}{\binom{d}{j}}$ and $p(x_1,\ldots, x_d)=f_i$ for all $(x_1,\ldots, x_d)$ such that $\sum_{j=0}^dx_j=i$. The mass function $p$ is the mass function of a $d$-dimensional Bernoulli random vector, which is exchangeabe by contruction. We have
\begin{equation}
\begin{split}
E[X_1]&=P(X_1=1)=\sum_{(x_1,\ldots,x_d):x_1=1}p(x_1,\ldots,x_d)=\sum_{i=1}^d\sum_{\substack{(x_1,\ldots,x_d):x_1=1,\\ \sum_{i=0}^dx_i=1}}p(x_1,\ldots,x_d)\\
&=\sum_{i=1}^d\sum_{\substack{(x_1,\ldots,x_d):x_1=1,\\ \sum_{i=0}^dx_i=1}}f_i=\sum_{i=1}^d\binom{d-1}{i-1}\frac{p_i}{\binom{d}{i}}\\
&=\sum_{i=1}^d\frac{(d-1)!}{(i-1)!(d-1-i+1)!}\frac{i!(d-i)!}{d!}p_i\\
&=\sum_{i=1}^d\frac{i}{d}p_i=\frac{1}{d}pd=p.
\end{split}
\end{equation}
Then $\XX\in \EEE(p)$.
\end{description}
Now let $S_d:=\sum_{i=1}^dX_i$. We have $P(S_d=j)=\binom{d}{j}f_j=p_j$ and $\{p_0,\ldots,, p_d\}\in \SSS(p)$.

\end{proof}

Therefore the three classes $\EEE(p)$, $\mathcal{S}_d(p)$ and $\DDD(dp)$ are essentially the same class, i.e.
\begin{equation}\label{map}
\begin{split}
\EEE(p)&\leftrightarrow \mathcal{S}_d(p)\equiv\DDD(dp)
\end{split}
\end{equation}
Thanks to the above proposition to find the generators of $\SSS(p)$ we can look for the generators of $\DDD(dp)$. This simplifies the search. The generators we find are in one-to-one relationship with the generators of $\EEE(p)$.

\section{Exchangeable Bernoulli generators \label{EG}}

We build on the results in \cite{fontana2018representation}, where the authors represent the Fr\'echet class of multivariate $d$-dimensional Bernoulli distributions with given margins and/or pre-specified moments as the points of a convex hull. The generators of the convex hull are mass functions in the class and they can be explicitly found. The range of application of this method is limited only by the computational effort required since the number of generators increases very quickly as the dimension increases. We show here that under the condition of exchangeability this limit can be overtaken because we analytically find the ray densities. As a consequence the dimension is no longer an issue. We focus on two classes: the class $\EEE(p)$ and  the class $\mathcal{E}_d(p,\rho )$, i.e.  the class of exchangeable Bernoulli vectors with given $p$ and given correlation $\rho$. The  one to one correspondence $E$ between the distributions $%
\pp_{S}\in \mathcal{S}_{d}(p)$ and $\ff_{p}\in \mathcal{E}_d(p)$  is also  a one-to-one correspondence between the distributions $%
\pp_{S}\in \mathcal{S}_{d}(p, \rho)$ and $\ff_{p}\in \mathcal{E}_d(p, \rho)$.

In Section 3.1 we represent the class $\mathcal{E}_d(p)$ as a convex hull of
mass functions in the class, which we call ray densities, so that each mass
function is a convex combinations of ray densities belonging to $\mathcal{E}_d
(p)$. We analytically find the ray densities and their number, that depends on the dimension $d$  and the mean value $p$.  The one-to-one map between $\EEE(p)$ and $\SSS(p)$ and Proposition \ref{d-eq-s}  are crucial.

In Section 3.2 we represent the class $\mathcal{E}_d(p,\rho )$, as well as $%
\mathcal{S}_{d}(p,\rho ),$ as a convex hull of ray densities.
We analyticall find them using the one-to-one correspondence between the
class $\mathcal{E}_d(p)$ and the class $\mathcal{S}_{d}(p)$ and between  the relative  subclasses $\mathcal{S}_{d}(p,\rho )$ and $\mathcal{E}_d(p,\rho )$.
 We prove that ray densities in $\mathcal{S}_{d}(p,\rho )$  have support on at most three points.
By so doing, also in this case the dimension $d$ is not an issue.

\subsection{For given marginal default probabilities}
Using the equivalence $\SSS(p)\equiv \DDD(pd)$ stated in  Proposition \ref{d-eq-s} a pmf in $\SSS(p)$  is a pmf on $\{0,\ldots,d\}$ with mean $pd$.
Thanks to the map $E$ in Equation \ref{map}  this is also equivalent to find a set of conditions that a pmf of a multivariate Bernoulli has to satisfy for being in $\EEE(p)$.  This fact is crucial in the following proposition.

\begin{proposition}\label{sd}
Let $\YY$ be a discrete random variable defined over $\{0,\ldots,d\}$ and let $p_Y$ be its pmf. Then
\begin{equation}
Y\in \SSS(p)\,\,\, \Longleftrightarrow\,\,\, \sum_{j=0}^d(j-pd)p_Y(j)=0.
\end{equation}

\end{proposition}

\begin{proof}
Let $\YY$ be a discrete random variable defined over $\{0,\ldots,d\}$. By Proposition \ref{d-eq-s} $Y\in \SSS(p)$ iff $E[Y]=pd$.
It holds
\begin{equation*}
E[Y]=pd \Longleftrightarrow E[Y-pd]=0 \Longleftrightarrow  \sum_{j=0}^d(j-pd)p_Y(j)=0.
\end{equation*}
\end{proof}
Using Proposition \ref{sd} we can find all generators of $\SSS(p)$ that, thanks to the map $E$ is equivalent to find all the generators of $\EEE(p)$.

We have to find the solutions $\pp_S=(p_0,\ldots, p_j)$ of
\begin{equation}\label{simp}
 \sum_{j=0}^d(j-pd)p_j=0.
\end{equation}
with the conditions $p_j\geq0, \, j=0,\ldots,d$ and $\sum_{j=0}^dp_j=1$.   From the standard theory of linear equations we know that all the positive solutions of  \ref{simp} are elements of the convex cone  \begin{equation}\label{cone}
\mathcal{C}_p=\{\zz\in \RR^{d+1}:  \sum_{j=0}^da_jz_j=0,\, I\zz\geq 0\},
 \end{equation}
 where $a_j=j-pd$ and $I$ is the $(d+1)\times(d+1)$ identity matrix, and therefore  can be generated as convex combinations of a set of generators which are referred to as extremal rays of the linear system. The proof of the following proposition follows Lemma 2.3 in \cite{terzer2009large}.

\begin{proposition}\label{multinulli}
 Let us consider the linear system
\begin{equation}\label{system}
A\zz=0, \zz\in \RR^{d+1}
\end{equation}
where $A$ is a $m\times (d+1)$ matrix, $m\leq d$ and $\rank A=m$. The extremal rays of the system \ref{system} have at most $m+1$ non-zero components.
\end{proposition}
\begin{proof}
 Let $\mathcal{C}_A=\{\zz\in \RR^{d+1}:  A\zz=0,\, I\zz\geq 0\}$ be the convex cone of all the positive solutions of \ref{system}.
 A solution $\rr$ of \ref{simp} is an extremal ray of $\mathcal{C}_A$ iff $I^*\zz=0$ for a submatrix $n_{I^*}\times(d+1)$, $I^*$ of $I$ and
\begin{equation}
\text{rank}\left[ \begin{split}
&A \\
 &I^*
\end{split}\right]=d.
\end{equation}

Therefore $\rank I^*\geq d-m$ and $\rr$ has at most $(d+1)-(d-m)=m+1$ non-zero components.
\end{proof}

\begin{corollary}\label{corp}
The extremal rays of the convex cone $\mathcal{C}_p$ in \ref{cone} have at most two non-zero components.
\end{corollary}
\begin{proof}
Let $a_j=j-pd$, $j=0,\ldots, d$. The matrix  $A=[a_0,\ldots, a_d]$ is the row vector of the coefficients. Since $\rank A=m=1$ then an extremal ray $\rr$ has at most two non-zero components.
\end{proof}

\begin{proposition}\label{binu}
The extremal rays of of the convex cone $\mathcal{C}_p$ in \ref{cone} are
\begin{equation}\label{binul}
p_{j_1,j_2}(y)=\left\{ \begin{array}{cc}
\frac{j_2-pd}{j_2-j_1}&y=j_1\\
\frac{pd-j_1}{j_2-j_1}&y=j_2\\
0&\text{otherwise}
\end{array}
\right.,
\end{equation}
with $j_1=0,1,\ldots, j_1^{M}$, $j_2=j_2^m, j_2^m+1, \ldots, d$, $j_1^M$ is the largest integer less that $pd$ and $j_2^m$ is the smallest integer greater than pd.

If $pd$ is integer  the extremal rays contain also
\begin{equation}\label{onenul}
p_{pd}(y)=\left\{ \begin{array}{cc}
1&y=pd\\
0&\text{otherwise}
\end{array}
\right..
\end{equation}

\end{proposition}

\begin{proof}
Let $a_j=j-pd$. Equation \ref{simp} becomes
\begin{equation}
\sum_{j=0}^da_jp_j=0.
\end{equation}
By Corollary \ref{corp} the extremal rays have at most two non zero components, say $j_1, j_2$. Therefore the extremal rays can be found considering the equations

\begin{equation}
a_{j_1}p_{j_1}+a_{j_2}p_{j_2}=0,
\end{equation}
where we make the non restrictive assumption $j_1<j_2$. The equation \ref{simp} has positive solutions only if $a_{j_1}a_{j_2}<0$. We observe that $a_{j_1}<0$ for $0\leq j_1 \leq j_1^M$ where $j_1^M$ is the largest integer less than $pd$ and $a_{j_2}>0$ for $j_2^m\leq j_2\leq d$ where $j_2^m$ is the smallest integer greater than $pd$. In this case we have $j_2^m=j_1^M+1$. It follows that for  $0\leq j_1 \leq j_1^M$  and $j_2^m\leq j_2\leq d$ we have $a_{j_1}a_{j_2}<0$. A positive solution of Equation \ref{simp} is

\begin{equation}
\left\{ \begin{array}{c}
\tilde{p}_{y}(j_1)=x_{j_1}=j_2-pd\\
\tilde{p}_{y}(j_2)=-x_{j_2}=pd-j_1
\end{array}
\right..
\end{equation}
We have $\tilde{p}_{y}(j_1)+\tilde{p}_{y}(j_2)=j_2-pd+pd-j_1=j_2-j_1$ and then the normalized extremal rays corresponding to $j_1$ and $j_2$ are given by \eqref{binul}.
If $pd$ is integer we have $a_{pd}=0$. It follows that \eqref{onenul} is also an extremal solution.

\end{proof}
We denote by $\Rjj$  and $R_{pd}$  the random variables whose pmf are $\rjj$ and $\rr_{pf}$ respectively.
We will refer to $\rjj$ and $\rr_{pf}$ as ray densities and
$\Rjj$  and $R_{pd}$ as ray random variables. Notice that $\rr_{(0,d)}=(1-p,0,\ldots,0, p)$.
\begin{corollary}\label{nary}
\begin{description}
\item If  $pd$ not integer there are $n_p=(j_1^m+1)(d-j_1^m)$ ray densities.
\item If $pd$ integer there are $n_p=d^2p(1-p)+1$ ray densities.
\end{description}
\end{corollary}
We have proved the following.
\begin{theorem}
The following holds.  $S_{d}\in \mathcal{S}_{d}(p )$ iff there exist $\lambda_1, \ldots, \lambda_{n_p}\geq 0$ summing up to 1 such that
\begin{equation}
\pp_{S}=\sum_{i=1}^{n_{p}}\lambda _{i}\rr_{i},
\end{equation}
where $\rr_{i}$ are the ray densities and $n_p$ is the number of ray densities.
\end{theorem}

\subsubsection{Second order moments}\label{secondmom}
Let $\XX\in \EEE(p)$ and let $\mu_2=E[X_iX_j]$ its second order cross moment.
\begin{proposition}
\label{prop_mu} Let $\boldsymbol{X}\in \mathcal{E}_d(p)$. It holds
\begin{equation}  \label{rho}
\mu_2=\sum_{k=0}^d\frac{k(k-1)}{d(d-1)}p_k.
\end{equation}
\end{proposition}

\begin{proof}
By exchangeability we can fix any pair $i, j\in\{1,\ldots, d\}$. It holds
\begin{equation*}
\begin{split}
\mu_2=P(X_i=1, X_j=1)&=\sum_{k=0}^dP(X_i=1, X_j=1|S_d=k)P(S_d=k)\\
&=\sum_{k=2}^d\frac{\binom{d-2}{k-2}}{\binom{d}{k}}p_k=\sum_{k=2}^d\frac{k(k-1)}{d(d-1)}p_k\\
&=\sum_{k=0}^d\frac{k(k-1)}{d(d-1)}p_k,
\end{split}
\end{equation*}
\end{proof}
Thanks to the one-to-one map $E$ we can find the bounds for the second order moments of $\EEE(p)$ using the second order moments of $S_d$.
 We have
\begin{equation}\label{smu2}
E[S_d^2]=E[(X_1+\ldots+X_d)^2]=pd+d(d-1)\mu_2.
\end{equation}
\begin{proposition}
Let $\XX\in \EEE(p)$. Then
if $pd$ is not integer
\begin{equation}
\frac{1}{d(d-1)}[-j_1^m(j_1^m+1)+2j_1^m pd]\leq \mu_2\leq p.
\end{equation}
If $pd$ is integer
\begin{equation}\label{smu2int}
\frac{p(pd-1)}{(d-1)}\leq \mu_2\leq p.
\end{equation}

\end{proposition}

\begin{proof}
From \eqref{smu2} we have $\mu_2=\frac{1}{d(d-1)}[E[S^2_d]-pd]$.  Since $S_d\in \SSS(p)$ its density is a convex linear combinations of the ray densities. It is known that the moments of $S_d$ are moments of the ray variables.
We obtain
\begin{equation}\label{mu2R}
E[\Rjj^2]=j_1^2\frac{j_2-pd}{j_2-j_1}+j_2^2\frac{j_2-pd}{pd-j_1}=-j_1j_2+(j_1+j_2)pd,
\end{equation}
and
\begin{equation}\label{mu2R2}
E[R_{pd}^2]=(pd)^2.
\end{equation}

To maximize $\mu_2$ we have to maximize $E[S_d^2]$. From \eqref{mu2R} and  \eqref{mu2R2} we easily get that the ray variable for which the second order moment is maximum is $R_{(0,d)}$ and we have $E[R_{(0,d)}^2]=(pd)^2$. Then, after some computations, $\mu_2^M=p$.

To minimize  $\mu_2$ we have to minimize $E[S_d^2]$. We consider two cases.

If $pd$ is not integer, from \eqref{mu2R} we have that the ray variable for which the second order moment is minimum is $R_{(j_1^M, j_2^m)}=R_{(j_1^M, j_1^M+1)}$, for which we have
$E[R_{(j_1^M, j_1^M+1)}^2]=-j_1^m(j_1^m+1)+(2j_1^m+1)pd$ and the assert follows.

If $pd$ is integer  the ray variable for which the second order moment is minimum is $R_{pd}$. Since $E[R_{pd}^2]=(pd)^2$, \eqref{smu2int} follows.

\end{proof}
Thanks to equation \eqref{eq:rho_e12}, the next corollary to the above proposition provides bounds for the correlation coefficient.
\begin{corollary}
Let $\XX\in \EEE(p)$. Then
if $pd$ is not integer
\begin{equation}
\frac{\frac{1}{d(d-1)}[-j_1^m(j_1^m+1)+2j_1^m pd]-p^2}{p(1-p)}\leq \rho\leq 1.
\end{equation}
If $pd$ is integer
\begin{equation}
-\frac{1}{d-1}\leq \rho\leq 1.
\end{equation}
\end{corollary}

\subsection{For given marginal default probabilities and default correlations}

In this section we consider the class of multivariate exchangeable Bernoulli
mass functions with given margins $p$ and given correlation $\rho $, i.e.
the class $\mathcal{E}_d(p,\rho )$.
We now find the generators of $\SSS(p, \rho)$.

Since  $S\in \SSS(p, \rho)$ iff $E[S_d]=pd$ and $E[S_d^2]=dp+d(d-1)\mu_2$, we can  define an homogeneous linear system whose solutions are the pmf in $\SSS(p, \rho)$.

\begin{proposition}
The following holds.  $S_{d}\in \mathcal{S}_{d}(p,\rho )$ iff there exist $\lambda_1, \ldots, \lambda_{n_p}\geq 0$ summing up to 1 such that
\begin{equation}
\pp_{S}=\sum_{i=1}^{n_{p}}\lambda _{i}\rr_{\rho, i},
\end{equation}
where $\rr_{ \rho, i}$ are the normalized extremal rays of the cone $\mathcal{C}_{p,\rho }$ defined by linear system:

\begin{equation}  \label{syst}
\left\{
\begin{array}{c}
\sum_{j=0}^d[{j}-pd]p_j=0 \\
\sum_{j=0}^d[j^2-(pd+d(d-1)\mu_2)]p_j=0.
\end{array}
\right .
\end{equation}
\end{proposition}

The following corollary of Proposition \ref{multinulli} characterizes the ray densities of $\SSS(p, \rho)$.


\begin{corollary}\label{trnu}
The extremal rays of $\mathcal{S}_d(p,\rho)$ have support on at most
three points.
\end{corollary}

\begin{proof}
The extremal rays of \eqref{simp} are the normalized extremal rays of the convex cone $\mathcal{C}=\{\zz\in \RR^{d+1}:  A\zz=0,\, I\zz\geq 0\}$, where $A$ is the matrix coefficients of \eqref{syst}.  We have $\rank A\leq 2$.
From Proposition \ref{multinulli} it follows $\text{rank} (\II^*)=d-3,\,d-2,\,d-1$  and $\II^*=(e_3,\ldots,e_n)^T$ to let  $\AA|\II$ have  $d-1$ independent rows. Since $\II^*\RRR=0$, if  $\text{rank} (\II^*)=d-3$, $\RRR$ has only three non zero components, if  $\text{rank} (\II^*)=d-2$, $\RRR$ has only two non zero components, and if  $\text{rank} (\II^*)=d-1$, $\RRR$ has only one non zero component. In the latter case all the mass is  one point.

\end{proof}

\begin{proposition}\label{trinulexpr}

The extremal rays of \eqref{simp} are $\rr_{\rho}=(p_0,\ldots, p_d)$, where $p_l=0, l\neq i,j,k$,
\begin{equation}
\begin{split}
&p_i=\frac{jk-(j+k-1)dp+d(d-1)\mu_{2}}{(k-i)(j-i)}\\
&p_j=-\frac{ik-(i+k-1)dp+d(d-1)\mu_{2}}{(k-j)(j-i)}\\
&p_k=\frac{ij-(i+j-1)dp+d(d-1)\mu_{2}}{(k-j)(k-i)},\\
\end{split}
\end{equation}
with $i<j<k$ and $p_i, p_j, p_k\geq0$
\end{proposition}
\begin{proof}
The extremal rays of \eqref{simp} can be found as follows. Let $\alpha_j:=j-pd$ and $\beta_j:=j^2-(pd+d(d-1)\mu_2)$, we can write system \eqref{syst} as follows:
\begin{equation}\label{syst2}
\left\{
\begin{array}{c}
\sum_{j=0}^d\alpha_jp_j=0
\\
\sum_{j=0}^d\beta_jp_j=0,
\end{array}
\right.
\end{equation}
Let now $A=\begin{bmatrix} \alpha_0& \ldots& \alpha_d\\ \beta_0 &\ldots& \beta_d\end{bmatrix}$. From Corollary \ref{trnu} we have to find the positive solutions $(z_j, z_j, z_k)$,  for $i<j<k$ , of

\begin{equation}\label{syst2}
\left\{
\begin{array}{c}
\alpha_ix_i+\alpha_jx_j+\alpha_kx_k=0
\\
\beta_ix_i+\beta_jx_j+\beta_kx_k=0,
\end{array}
\right.
\end{equation}
Then, from a positive solution, we find $p_l=\frac{z_l}{z_i+z_j+z_k}$, $l\in \{i,j,k\}$.

Letting $x_k=1$ the system \ref{syst2} becomes

\begin{equation}\label{syst2}
\left\{
\begin{array}{c}
\alpha_ix_i+\alpha_jx_j=-\alpha_k\\
\beta_ix_i+\beta_jx_j=-\beta_k,
\end{array}
\right.
\end{equation}
and its solution can be  determined by standard computation using Cramer's formula.
\end{proof}
%
We conclude this section with the following proposition that gives necessary and sufficient conditions for a ray density in $\SSS(p)$ to be also a ray density in $\SSS(p, \rho)$ .

\begin{proposition}
A ray density $\rr\in \SSS(p, \rho)$  has support on two points iff it is a ray density in $\SSS(p)$ and $\mu_2^{\rr}=\mu_2$, where $\mu_2^{\rr}$ is the second order cross moment of $\rr$.

\end{proposition}

\begin{proof}
If $\rr$ is a solution of \eqref{syst} it is also a solution of \eqref{simp} and since it has support of two poins by assumption it is an extremal solution. Thus $\rr\in \SSS(p)$ is a ray density. Viceversa if $\rr\in \SSS(p)$ it satisfies the first equation of \eqref{syst} by definition  and if $\mu_2^{\rr}=\mu_2$ it also satisfy the second equation by construction.
Since it has mass on two points it is an extremal solution of \eqref{syst}.
\end{proof}
\section{Financial risk measures and their bounds}\label{FM}

As measures of portfolio risk we consider the value at risk (VaR) and the
expected shortfall (ES) of $S_d$. We recall their definition for a general random variable $Y$.

\begin{definition}
Let $Y$ be a random variable representing a loss with finite mean. Then the
$ \VaR$ at level $\alpha $ is defined by
\begin{equation}
 \VaR(Y)=\inf \{y\in \mathbb{R}:P(Y\leq y)\geq \alpha \}
\end{equation}%
and the expected shortfall at level $\alpha $ is defined by
\begin{equation}
 \ES(Y)=E[Y|Y\geq Var_{\alpha }(Y)]
\end{equation}
\end{definition}

The following proposition provides the bounds for the $ \VaR$ and $\ES$ of $S_d$, for $S_d$ in a given class.

\begin{proposition}
\begin{enumerate}
\item Let $S_d\in \mathcal{S}_d(p) [\SSS(p, \rho)]$ and let $\VaR(S_d)$ be its value at risk. Then $$\min_{R}  \VaR(R)\leq  \VaR(S_d)\leq \max_{R}  \VaR(R),$$ where $R$ are the ray densities of $\mathcal{S}_d(p) [\SSS(p, \rho)]$.

\item Let $S_d\in \mathcal{S}_d(p)[\SSS(p, \rho)]$ and let $ES_{\alpha}(S_d)$ be its expected
shortfall. Then $$\min_R  \VaR(R)\leq \ES(S_d)\leq d,$$ where $R$ are the ray densities of $\mathcal{S}_d(p) [\SSS(p, \rho)]$.
\end{enumerate}
\end{proposition}
\begin{proof}
\begin{enumerate}
\item Let $\tau_S=\VaR(S_d)=inf \{y\in \TT: P(S_d\leq y)\geq \alpha\}$. Let $\tau_i=\VaR(R_i)$, $\tau_M=\max_i \tau_i$ and $\tau_m=\min_i \tau_i$. It holds
\begin{equation}
\begin{split}
P(S_d\leq \tau_M)&=\sum_{y\leq t_M}p_S(y)=\sum_{y\leq t_M}\sum_{i=1}^{n_p}\lambda_ip_{R_i}(y)\\
&=\sum_{i=1}^{n_p}\lambda_i\sum_{y\leq t_M}p_{R_i}(y)\geq\sum_{i=1}^{n_p}\lambda_i \alpha=\alpha,
\end{split}
\end{equation}
thus $\tau_S\leq \tau_M$.
It holds
\begin{equation}
\begin{split}
P(S_d\leq \tau_m)&
=\sum_{i=1}^{n_p}\lambda_i\sum_{y\leq t_m}p_{R_i}(y)=\sum_{i=1}^{n_p}\lambda_i \beta_i,
\end{split}
\end{equation}
with $\beta_i\leq \alpha$ therefore we have $P(S_d\leq \tau_m)\leq \alpha$. Thus $\tau_S\geq \tau_m$ and $\tau_m\leq \tau_S\leq \tau_M$.
\item $ES_{\alpha}\geq \tau_m$ and $ES_{\alpha}\leq d$ are trivial.
\end{enumerate}

\end{proof}

The above propositions shows that  $ \VaR$   reaches the maximum and minimum values in $\mathcal{S}_d(p)$ [$\mathcal{S}_d(p, \rho)$]  on the ray densities and
therefore we are able to explicitly find them.

\begin{remark}\label{ESbounds}
The bounds for $ES_{\alpha}$ are weaker and trivial. Nevertheless, at least in some cases, they cannot be improved.  In fact, consider the ray density $\rr=(1-p,\ldots, p)\in \EEE(p)$. If $1-p\leq  \alpha$ then $ES_{\alpha}=d$. As a consequence for marginal default probabilities higher then $1-\alpha$ the bound is reached.

\end{remark}

Thanks to Proposition \ref{binu} that gives the analytical expression of the ray densities of $\SSS(p)$, the following proposition provides the analytical bounds for $\VaR$  in $\SSS(p)$.
\begin{proposition}
Let us consider the class $\SSS(p)$ and 
let $j_1^p=\frac{(p-(1-\alpha))d}{\alpha}$.
\begin{enumerate}
\item  If $j_1^p<0$, $\min \VaR(\Rjj)=0$ and $\max \VaR(\Rjj)=j_2^*$, where $j_2^*$ is the largest integer smaller than $\frac{pd}{ 1-\alpha}$.

\item  If $0\leq j_1^p\leq j_1^M$, $\min  \VaR(\Rjj)=j_1^*$, where $j_1^*$ is the smallest  integer greater or equal to  $j_1^p$  and $\max \VaR(\Rjj)=d$.

\item  If $j_1^p>j_1^M$, $\min \VaR(\Rjj)=j_2^m=j_1^M+1$ and $\max \VaR(\Rjj)=d$.
In this case, if  $pd$ is integer $j_1^M+1=pd$.
\end{enumerate}
\end{proposition}
\begin{proof}
Let us consider first the case $pd$ not integer.
The ray densities are given in \eqref{binul} with $0\leq j_1\leq j_1^M$ and $j_1^M+1\leq j_2\leq d$. From the definition of $VaR$ we have
\begin{equation}
\VaR(\Rjj)=j_1 \Longleftrightarrow \rjj(j_1)\geq\alpha.
\end{equation}
It follows
\begin{equation}
\begin{split}
&\rjj(j_1)=\frac{j_2-pd}{j_2-j_1}\geq \alpha,\\
\end{split}
\end{equation}
then
\begin{equation}
\begin{split}
&{j_2}\geq-\frac{ \alpha}{1-\alpha}{j_1}+ \frac{pd}{1-\alpha}.
\end{split}
\end{equation}
We also know that $j_2\leq d$, so let us determine the point $P$ of intersection of ${j_2}=-\frac{ \alpha}{1-\alpha}{j_1}+ \frac{pd}{1-\alpha}$ and $j_2=d$.
The solution of

\begin{equation}
\left\{ \begin{array}{l}
{j_2}=-\frac{ \alpha}{1-\alpha}{j_1}+ \frac{pd}{1-\alpha}\\
j_2=d
\end{array}
\right.
\end{equation}
is $P=(j_1^P, j_2^P)=(\frac{(p-(1-\alpha))d}{\alpha}, d)$.
We distinguish three cases, depending on $j_1^P$.
\begin{enumerate}
\item{$j_1^p<0$}. In this case it will follow that $ \VaR(R_{(0,j_2)})=0$ for all $\frac{pd}{1-\alpha}<j_2\leq d$ and then the minimum value of $ \VaR(\Rjj)=0$. With respect to the maximum value of $ \VaR(\Rjj)$ it will be obtained by $\VaR(R_{(0,j_2^*)})$, where $j_2^*$ is the largest integer smaller than $\frac{pd}{1-\alpha}$.

\item{$0\leq j_1^p\leq j_1^M$.} Let us define $j_1^*$ as the smallest integer greater or equal to $j_1^P$. It follows that $\VaR(R_{(j_1^*,j_2)})=j_1^*$, $j_2^*<j_2\leq d$ with $j_2^*$ is the smallest integer greater or equal to  $-\frac{ \alpha}{1-\alpha}{j_1^*}+ \frac{pd}{1-\alpha}$. Then the minimum value of $VaR_{\alpha}(R_{(j_1,j_2)})=j_1^*$. The maximum value of $ \VaR(R_{(j_1,j_2)})$ is $d$.

\item{$j_1^p>j_1^M$.} In this case $\VaR (R_{(j_1,j_2)})=j_2$. Then the minimum value of  $\VaR(R_{(j_1,j_2)})=j_2^m=j_1^M+1$ and the maximum value of  $ \VaR(R_{(j_1,j_2)})=d$. If $pd$ is integer we also have $j_1^M+1=pd$.
\end{enumerate}
\end{proof}

We can also explicitly find the bounds in $\SSS(p,\rho)$ by searching the maximum e minimum $\VaR$ among the ray densities, whose analytical expression is given in Proposition \ref{trinulexpr}.  The analytical computation of $\VaR$ is out of the aim of the present paper, here we simply  serach for the minimum $\VaR$ and the maximum $\VaR$ among the ray densities.

\section{Model risk analysis\label{MR}}

The theory developed so far allows us to perform model risk analysis.

Consistently with it, let us suppose we have a credit portfolio $P$ with 100
obligors. Let the random vector $\boldsymbol{X}=(X_{1},\ldots ,X_{100})$
collect the default indicators for the portfolio $P$ and assume $\boldsymbol{%
X}\in \mathcal{E}$, where $\mathcal{E}:=\mathcal{E}_{100}$. The variable $S:=S_{100}$ represents the number of
defaults and the distribution of $S$ represents the distribution of the
loss. We analytically find bounds of $ \VaR$ and $\ES$, for
$\alpha =0.90$, $\alpha =0.95$ and $\alpha =0.99$ for two classes of
multivariate exchangeable Bernoulli variables $\mathcal{E}(p)$ and  $\mathcal{E}(p,\rho )$.

The analysis of these two classes of models allows us to study the two
aspects of model risk mentioned in the Introduction, the risk associated to
the pure choice of a "wrong" model (pure model risk) and the one associated
to a "wrong" calibration of the joint model, through default correlation
(calibration risk). In both cases we do not investigate the correctness of
the marginal default probability, which would be the case if we were
investigating marginal model risk.

The bounds of the first class provide an economically sensible measure of
pure joint model risk. To complete the picture,  for any $p$ we provide
the range of admissible correlations for the hundred Bernoulli variables.

The bounds of the second class provide a measure of calibration risk. The
bounds are obtained for a specific correlation coefficient: we perform a
sensitivity analysis of their behavior when $\rho $ changes. \ For each
correlation, we also consider $\VaR$ and $\ES$ associated
to a specific joint model (the Bernoulli mixture one), to show how the
method can be used to assess not calibration risk in general, but the
calibration  risk of a specific model, considering how far its $\VaR
$ and $\ES$ are from the bounds.

In all cases we consider three scenarios corresponding to three marginal
default probabilities $p=0.3\%$, $p=1.7\%$ and $p=26.6\%$, which are the
1-year marginal default probabilities resulting from \cite{SeP}  table 13 page
40,  for the rating
classes $A,BBB$ and $B$.
%
\subsection{Pure model risk\label{ex}}

Here we deal with $\mathcal{E}(p)$ in the three scenarios $%
p=0.3\%,p=1.7\%$ and $p=26.6\%.$  All the results in this section are  analytical.

\subsubsection{Scenario 1: $p=0.3\%$}

Before computing $\VaR$  and $\ES$ for the class $\mathcal{S%
}(0.3\%),$ corresponding to Moody's A rating, let us describe it. The class  has $100$ ray densities that we can find analitically and we found that
all ray densities have different correlations. The bounds for the all moments of the distributions in the class  are reached on the ray densities as proved in \cite{fontana2018representation}.  In this case the bounds for the second order moment and correlation are analytical, as proved in  Section \ref{secondmom}. The moments up to order four and correlation are given in Table \ref{tabella_mom_100_3_997}. Obviously, the
first moment coincides with \thinspace $p$ and its range is a singleton.  Notice that  all positive correlations and  some
negative  are allowed. This is possible since we consider finite sequences of Exchangeable Bernoulli variables and not only the mixing models, i.e. the De Finetti's sequences. So, per se, independently of any model, a hundred
Bernoulli default indicators with equicorrelation cannot span negative
dependence, but are able to span any level of positive dependence and zero
correlation.
\begin{table}[h!]
\centering
{\footnotesize \
\begin{tabular}{lcc}
Order & Min moment & Max moment \\ \hline
1 & 0.003 & 0.003 \\
2 & 0 & 0.003 \\
3 & 0 & 0.003 \\
4 & 0 & 0.003 \\
$\rho$ & -0.003 & 1 \\ \hline
\end{tabular}
}
\caption{Moments $\mathcal{E}(0.3\%)$ class of multivariate Bernoulli }
\label{tabella_mom_100_3_997}
\end{table}

Table  \ref%
{tabella_var_100_3_997} shows the bounds for  the $\VaR$  for the
three levels $\alpha =0.90$, $\alpha =0.95$ and $\alpha =0.99$.

\begin{table}[h!]
\centering
{\footnotesize \
\begin{tabular}{lcc}
Quantile & Min $\VaR$  & Max $\VaR$  \\ \hline
0.9 & 0 & 2 \\
0.95 & 0 & 5 \\
0.99 & 0 & 29 \\ \hline
\end{tabular}
}
\caption{$\VaR$  of the number of defaults for the $\mathcal{E}(0.3\%)$ class of
multivariate Bernoulli }
\label{tabella_var_100_3_997}
\end{table}

Table \ref{tabella_es_100_3_997} shows the bounds for the ES on the ray densities for the
three levels $\alpha =0.90$, $\alpha =0.95$ and $\alpha =0.99$.

\begin{table}[h!]
\centering
{\footnotesize \
\begin{tabular}{lcc}
Quantile & Min ES & Max ES \\ \hline
0.9 & 0.3 & 2 \\
0.95 & 0.3 & 5 \\
0.99 & 0.3 & 29 \\ \hline
\end{tabular}
}
\caption{ES of the number of defaults for the $\mathcal{E}(0.3\%)$ class of
multivariate Bernoulli }
\label{tabella_es_100_3_997}
\end{table}

\subsubsection{ Scenario 2 }

Let us assume $p=1.7\%$, The class $\mathcal{S}(1.7\%)$ has $198$ ray
distributions of $S$ with $198$ different correlations. Table \ref{tabella_mom_100_17_983} provides the bound of the moments also for this class.

\begin{table}[h!]
\centering
{\footnotesize \
\begin{tabular}{lcc}
Order & Min moment & Max moment \\ \hline
1 & 0.017 & 0.017 \\
2 & 0 & 0.017 \\
3 & 0 & 0.017 \\
4 & 0 & 0.017 \\
$\rho$ & -0.009 & 1 \\ \hline
\end{tabular}
}
\caption{Moments $\mathcal{E}(1.7\%)$ class of multivariate Bernoulli }
\label{tabella_mom_100_17_983}
\end{table}
 Table \ref{tabella_var_100_17_983} shows the bounds for  the $\VaR$  for the
three levels $\alpha =0.90$; $\alpha =0.95$ and $\alpha =0.99$.

\begin{table}[h!]
\centering
{\footnotesize \
\begin{tabular}{lcc}
Quantile & Min $\VaR$  & Max $\VaR$  \\ \hline
0.9 & 0 & 16 \\
0.95 & 0 & 33 \\
0.99 & 1 & 100 \\ \hline
\end{tabular}
}
\caption{$\VaR$  of the number of defaults for the $\mathcal{E}(1.7\%)$ class of
multivariate Bernoulli }
\label{tabella_var_100_17_983}
\end{table}

Table \ref{tabella_es_100_17_983} shows  the bounds for the $ES$ on the ray densities for the
three levels $\alpha =0.99$; $\alpha =0.95$ and $\alpha =0.90$. Since $1.7\%\geq 1\%$ we have $\text{ES}_{0.99}=100$, as noticed in Remark \ref{ESbounds}.

\begin{table}[h!]
\centering
{\footnotesize \
\begin{tabular}{lcc}
Quantile & Min ES & Max ES \\ \hline
0.9 & 1.7 & 16 \\
0.95 & 1.7 & 33 \\
0.99 & 1.7 & 100 \\ \hline
\end{tabular}
}
\caption{ES of the number of defaults for the $\mathcal{E}(1.7\%)$ class of
multivariate Bernoulli }
\label{tabella_es_100_17_983}
\end{table}

\subsubsection{Scenario 3}

We consider the class $\mathcal{E}(26.6\%)$. The number of ray densities is
much higher relative to the other two classes considered since it is 1998. Table \ref{tabella_mom_100_133_367} shows that the range of the third and fourth moments of this class is wider that for the other classes.

\begin{table}[h!]
\centering
{\footnotesize \
\begin{tabular}{lcc}
Order & minmom & maxmom \\ \hline
1 & 0.266 & 0.266 \\
2 & 0.069 & 0.266 \\
3 & 0.017 & 0.266 \\
4 & 0.004 & 0.266 \\
$\rho$ & -0.01 & 1 \\ \hline
\end{tabular}
}
\caption{Moments $\mathcal{E}(26.6\%)$ class of multivariate Bernoulli }
\label{tabella_mom_100_133_367}
\end{table}

Table  \ref%
{tabella_var_100_133_367} shows the bounds for  the $\VaR$  for the
three levels $\alpha =0.90$; $\alpha =0.95$ and $\alpha =0.99$.

\begin{table}[h!]
\centering
{\footnotesize \
\begin{tabular}{lcc}
Quantile & Min $\VaR$ & Max $\VaR$  \\ \hline
0.9 & 19 & 100 \\
0.95 & 23 & 100 \\
0.99 & 26 & 100 \\ \hline
\end{tabular}
}
\caption{$\VaR$ of the number of defaults for the $\mathcal{E}(26.6\%)$ class
of multivariate Bernoulli }
\label{tabella_var_100_133_367}
\end{table}
%

The following Table \ref{tabella_es_100_133_367} shows the bounds for the $\ES$ on the ray densities for the
three levels $\alpha =0.90$; $\alpha =0.95$ and $\alpha =0.99$. As one can see the maximum $\ES$ is d=100 for each $\alpha$, in fact $26.6\%\geq 1\%$.
\begin{table}[h!]
\centering
{\footnotesize \
\begin{tabular}{lcc}
Quantile & Min ES & Max ES \\ \hline
0.9 & 26.6 & 100 \\
0.95 & 26.6 & 100 \\
0.99 & 26.6 & 100 \\ \hline
\end{tabular}
}
\caption{ES of the number of defaults for the $\mathcal{E}(26.6\%)$ class of
multivariate Bernoulli }
\label{tabella_es_100_133_367}
\end{table}

\subsubsection{Cross scenario comparisons}

The reader can appreciate how model risk increases, when the marginal
probability does, and  when the risk measure is $\VaR$, looking at Figure \ref{VAR0}. The computation permits to conclude that the VaR range increases with the marginal default  probability, and not only  with the level of confidence (which is the standard result). Also, both the minimum and the maximum are non decreasing with p.

\subsection{Calibration risk}

In this Section we examine the behavior of the loss under the three
scenarios above for the marginal default probability, when, on top of the
marginal, a specific value of the equicorrelation has been selected. We deal
with $\mathcal{E(}p,\rho )$ in the three scenarios $p=0.3\%,p=1.7\%$ and $p=26\%$ and provide bounds for $\VaR$  for three levels of correlation: $\rho=\frac{1}{6}; \frac{1}{2}; \frac{5}{6}$. 
Here, the ray densities are analytical as well as their VaR.  The bounds are found by computationally searching the maximum and minimum VaR   among the ray densities. 

As a benchmark we choose an exchangeable Bernoulli mixing model from the credit risk literature. We estimate the  $\beta$-mixing model of each scenario and compute its $\VaR$.
Let $S_{\beta}$  be the number of default of the $\beta$-mixing models, we have (for a complete overview see \cite{mcneil2005quantitative}):

\begin{equation}
p_{\beta}(j)=\binom{d}{j}\int_0^1p^k(1-p)^{d-k}d\Psi(p), 
\end{equation}
where  $\Psi\sim \beta(a,b)$ the mixing variable. We have 
\begin{equation}
\begin{split}
&p=E[\Psi]\\
&\mu_2=E[\Psi^2].
\end{split}
\end{equation}
Therefore we estimate the $\beta$ parameters $a$ and $b$ by
\begin{equation}
\begin{split}
&p=\frac{a}{a+b}\\
&\mu_2=\frac{a(a+1)}{(a+b)(a+b+1)}.
\end{split}
\end{equation}
Notice that for this model $\rho=0$ is not admissible.
\subsubsection{Scenario 1}

Table 
\ref{tabella_VAR_100_3_997-1_6}
provides the bounds of $\VaR$ when only correlation is known and it
is $\frac{1}{6}$ and the corresponding measures for the $\beta$-mixing
model.
\begin{table}[h!]
\centering
{\footnotesize \
\begin{tabular}{lccc}
Quantile & min$\VaR$  & max $\VaR$  & $\beta$ -$\VaR$  \\ \hline
0.9 & 0 & 2 & 0 \\
0.95 & 0 & 5 & 0 \\
0.99 & 1 & 22 & 9 \\ \hline
\end{tabular}
}
\caption{$\VaR$ of the number of defaults for the $\mathcal{E}(0.3\%, \frac{1}{6%
})$ class of multivariate Bernoulli }
\label{tabella_VAR_100_3_997-1_6}
\end{table}

Table 
\ref{tabella_VAR_100_3_997-1_2}
provides the bounds of  $\VaR$ when only correlation is known and it
is $\frac{1}{2}$ and the corresponding $\VaR$ for the $\beta$-mixing
model.

\begin{table}[h!]
\centering
{\footnotesize \
\begin{tabular}{lccc}
Quantile & min $\VaR$  & max $\VaR$  & $\beta$ -$\VaR$  \\ \hline
0.9 & 0 & 1 & 0 \\
0.95 & 0 & 3 & 0 \\
0.99 & 0 & 21 & 4 \\ \hline
\end{tabular}
}
\caption{$\VaR$  of the number of defaults for the $\mathcal{E}(0.3\%, \frac{1}{2%
})$ class of multivariate Bernoulli }
\label{tabella_VAR_100_3_997-1_2}
\end{table}

Table 
\ref{tabella_VAR_100_3_997-5_6}
provides the bounds of $\VaR$  when correlation is known and it
is $\frac{5}{6}$ and the corresponding measure for the $\beta$-mixing
model.

\begin{table}[h!]
\centering
{\footnotesize \
\begin{tabular}{lccc}
Quantile & min $\VaR$  & max $\VaR$  & $\beta$ -$\VaR$  \\ \hline
0.9 & 0 & 0 & 0 \\
0.95 & 0 & 1 & 0 \\
0.99 & 0 & 7 & 0 \\ \hline
\end{tabular}
}
\caption{$\VaR$  of the number of defaults for the $\mathcal{E}(0.3\%,\frac{5}{6}%
)$ class of multivariate Bernoulli }
\label{tabella_VAR_100_3_997-5_6}
\end{table}

\subsubsection{Scenario 2}

Table 
\ref{tabella_VAR_100_17_983-1_6}
provides the bounds of  $\VaR$  when  correlation is known and it
is $\frac{1}{6}$ and the  the $\beta$-mixing
model $\VaR$.

\begin{table}[h!]
\centering
{\footnotesize \
\begin{tabular}{lccc}
Quantile & min $\VaR$  & max $\VaR$  & $\beta$ -$\VaR$  \\ \hline
0.9 & 0 & 16 & 5 \\
0.95 & 1 & 25 & 11 \\
0.99 & 2 & 55 & 29 \\ \hline
\end{tabular}
}
\caption{$\VaR$  of the number of defaults for the $\mathcal{E}(1.7\%, \frac{1}{6%
})$ class of multivariate Bernoulli }
\label{tabella_VAR_100_17_983-1_6}
\end{table}

Table 
\ref{tabella_VAR_100_17_983-1_2}
provides the bounds of  $\VaR$   when  correlation is known and it
is $\frac{1}{2}$ and the corresponding measure for the $\beta$-mixing
model.
%


\begin{table}[h!]
\centering
{\footnotesize \
\begin{tabular}{lccc}
Quantile & min$\VaR$  & max $\VaR$ & $\beta$ -$\VaR$  \\ \hline
0.9 & 0 & 9 & 0 \\
0.95 & 0 & 25 & 5 \\
0.99 & 1 & 93 & 57 \\ \hline
\end{tabular}
}
\caption{$\VaR$  of the number of defaults for the $\mathcal{E}(1.7\%, \frac{1}{2%
})$ class of multivariate Bernoulli }
\label{tabella_VAR_100_17_983-1_2}
\end{table}

Table 
\ref{tabella_VAR_100_17_983-5_6}
provides the bounds of $\VaR$  when  correlation is known and it
is $\frac{5}{6}$ and the corresponding $\VaR$ for the $\beta$-mixing
model.

\begin{table}[h!]
\centering
{\footnotesize \
\begin{tabular}{lccc}
Quantile & min $\VaR$  & max $\VaR$  & $\beta$ -$\VaR$  \\ \hline
0.9 & 0 & 3 & 0 \\
0.95 & 0 & 8 & 0 \\
0.99 & 61 & 100 & 94 \\ \hline
\end{tabular}
}
\caption{$\VaR$  of the number of defaults for the $\mathcal{E}(1.7\%,\frac{5}{6}%
)$ class of multivariate Bernoulli }
\label{tabella_VAR_100_17_983-5_6}
\end{table}

\subsubsection{Scenario 3}

Table
\ref{tabella_VAR_100_133_367-1_6}
provides the bounds of  $\VaR$  when  correlation is known and it
is $\frac{1}{6}$ and the corresponding measures for the $\beta$-mixing
model.
In this case the number of generators of the class significantly increases.
In fact, the class $\mathcal{E}(26.6\%, \frac{1}{6})$ is generated by 32.372
ray densities.

\begin{table}[h!]
\centering
{\footnotesize
\begin{tabular}{lccc}
Quantile & min $\VaR$  & max $\VaR$  & $\beta$ -$\VaR$  \\ \hline
0.9 & 21 & 82 & 53 \\
0.95 & 26 & 100 & 62 \\
0.99 & 38 & 100 & 76 \\ \hline
\end{tabular}
}
\caption{$\VaR$  of the number of defaults for the $\mathcal{E}(26.6\%, \frac{1}{%
6})$ class of multivariate Bernoulli }
\label{tabella_VAR_100_133_367-1_6}
\end{table}

Table 
\ref{tabella_VAR_100_133_367-1_2}
provide the bounds of $\VaR$ when  correlation is known and it
is $\frac{1}{2}$ and the corresponding measure for the $\beta$-mixing
model.
\begin{table}[h!]
\centering
{\footnotesize \
\begin{tabular}{lccc}
Quantile & min $\VaR$  & max $\VaR$  & $\beta$ -$\VaR$ R \\ \hline
0.9 & 42 & 100 & 82 \\
0.95 & 56 & 100 & 93 \\
0.99 & 63 & 100 & 100 \\ \hline
\end{tabular}
}
\caption{$\VaR$  of the number of defaults for the $\mathcal{E}(26.6\%, \frac{1}{%
2})$ class of multivariate Bernoulli }
\label{tabella_VAR_100_133_367-1_2}
\end{table}

Table 
\ref{tabella_VAR_100_133_367-5_6}
provide the bounds of  $\VaR$  when  correlation is known and it
is $\frac{5}{6}$ and the corresponding measure for the $\beta$-mixing
model.

\begin{table}[h!]
\centering
{\footnotesize \
\begin{tabular}{lccc}
Quantile & min $\VaR$  & max $\VaR$  & $\beta$ -$\VaR$  \\ \hline
0.9 & 81 & 100 & 100 \\
0.95 & 86 & 100 & 100 \\
0.99 & 88 & 100 & 100 \\ \hline
\end{tabular}
}
\caption{$\VaR$  of the number of defaults for the $\mathcal{E}(26.6\%, \frac{5}{%
6})$ class of multivariate Bernoulli }
\label{tabella_VAR_100_133_367-5_6}
\end{table}
\subsubsection{Cross scenario comparisons}
Figures  \ref{VAR1}, \ref{VAR2} and \ref{VAR3} plot the bounds for VaR when  $\rho $ takes equispaced value in the range $[0, \frac{11}{12})$.
The reader can appreciate how calibration risk increases, when the marginal
probability and the correlation does. It also emerges that the $\VaR$ of the $\beta$-mixing model sometimes reaches the bound and depending on $p$ and $\rho$ its values with respect to the bounds  significantly change. In particular for low $p$ the $\VaR$ of the $\beta$-mixing model coincides with the minimum $\VaR$.
The plots show that, even if the $\beta$-mixing model is calibrated to match the moments of the Bernoulli, it tends to produce a VaR close to the minimum one for low $p$, and close to the maximum for high $p$. In any case, the width of the band between the minimum and the maximum, together with the specific location of the $\beta$ VaR within it, give a sense of how wrong the risk appreciation can go, when calibrating a specific correlation, and how stringent is the choice of a specific multivariate distribution within that calibration.

\newpage

\section{Conclusions\label{fine}}

Measuring model risk in credit and default modeling is important, at least
to have a sense of the consequences of mispricing of financial products,
forecasting errors etc. Since, at present, model risk in credit and default
cannot be avoided, we can try to measure it. This paper does exactly that,
in a very general context (exchangeable, equicorrelated Bernoulli), using
two popular risk measures, VaR and ES.
The main contributions are the closed form results for the VaR bounds and the moments of the multivariate distributions, as well as the numerical examples which show how big model risk can be, with a portfolio of 100 obligors with equal exposure, especially when the marginal default probability is high.

\newpage

\bibliographystyle{ieeetr}
\bibliography{biblio}

\begin{figure}[btp]\caption{VAR ranges  for $p=0.03\%; 1.7\%; 26.6\%$ } \label{VAR0}
\centering
\begin{subfigure}[b]{0.3\textwidth}
\includegraphics[width=\textwidth]{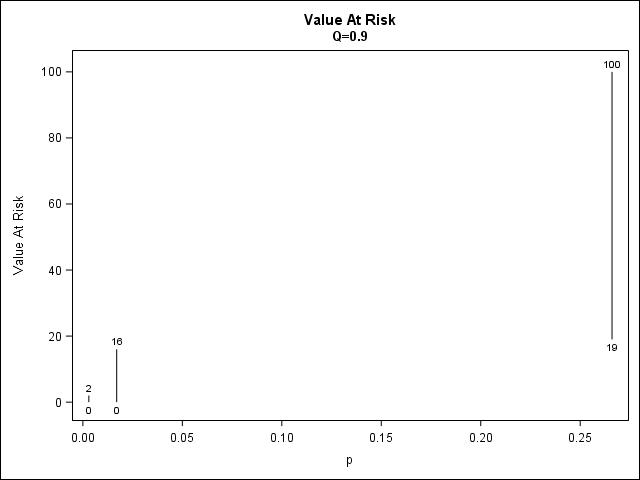}
\end{subfigure}
\begin{subfigure}[b]{0.3\textwidth}
\includegraphics[width=\textwidth]{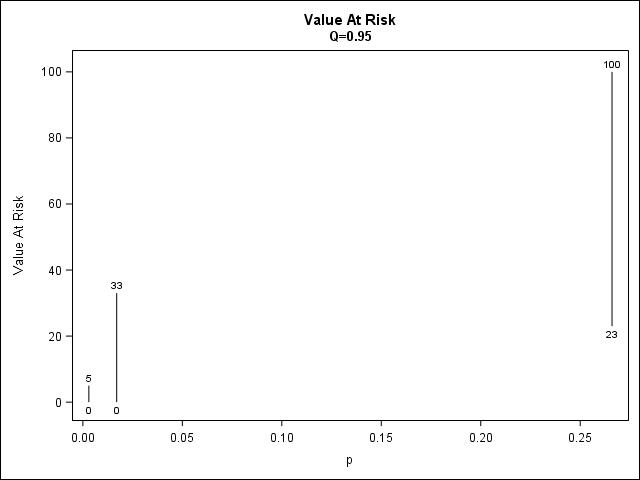}
\end{subfigure}
\begin{subfigure}[b]{0.3\textwidth}
\includegraphics[width=\textwidth]{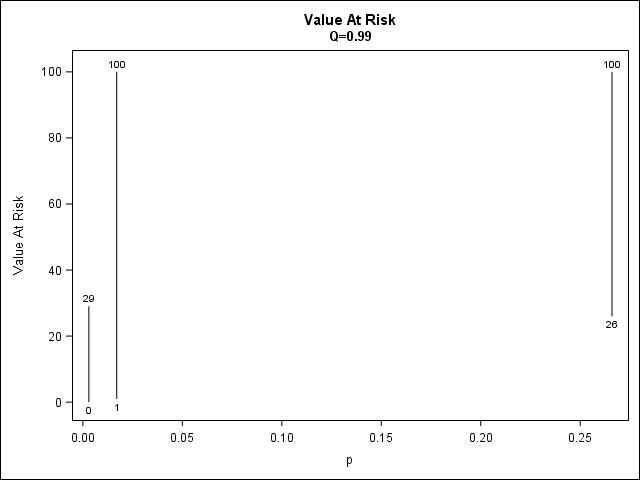}
\end{subfigure}
\end{figure}

\begin{figure}[h!]\caption{VAR  bounds  for $p=0.03\%$ and different $\rho$ and $\beta$-mixing model VAR}\label{VAR1}
\centering
\begin{subfigure}[b]{0.3\textwidth}
\includegraphics[width=\textwidth]{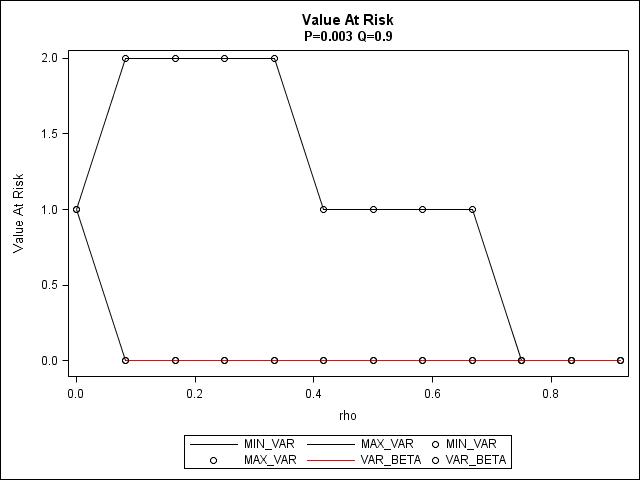}
\end{subfigure}
\begin{subfigure}[b]{0.3\textwidth}
\includegraphics[width=\textwidth]{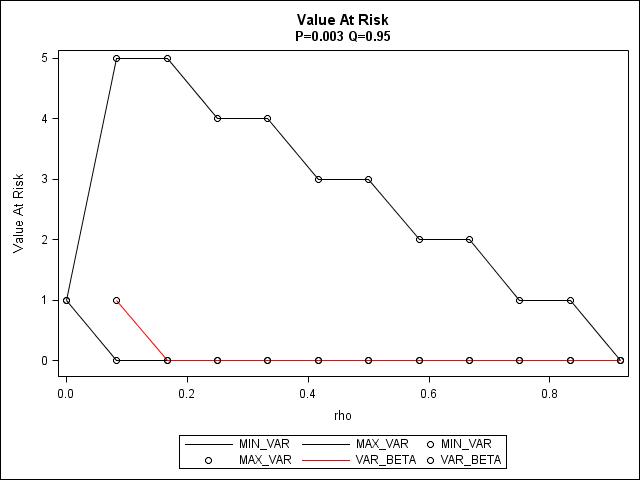}
\end{subfigure}
\begin{subfigure}[b]{0.3\textwidth}
\includegraphics[width=\textwidth]{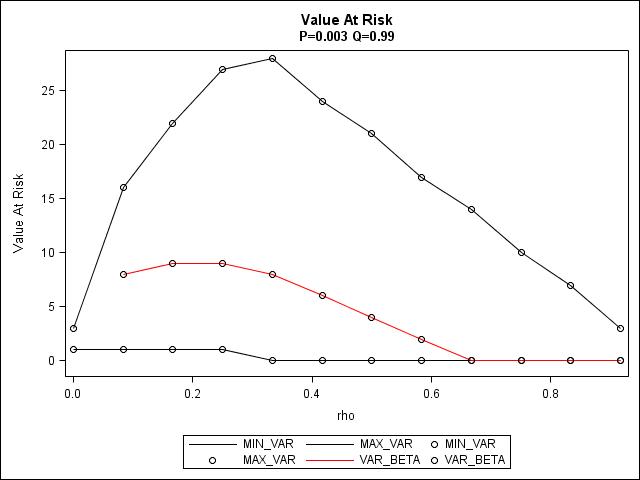}
\end{subfigure}
\end{figure}

\begin{figure}[h!]\caption{VAR  bounds  for $p=1.7\%$ and different $\rho$ and $\beta$-mixing model VAR}\label{VAR2}
\centering
\begin{subfigure}[b]{0.3\textwidth}
\includegraphics[width=\textwidth]{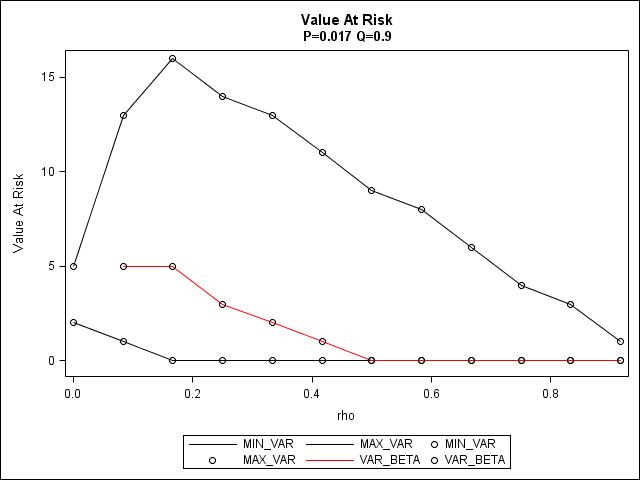}
\end{subfigure}
\begin{subfigure}[b]{0.3\textwidth}
\includegraphics[width=\textwidth]{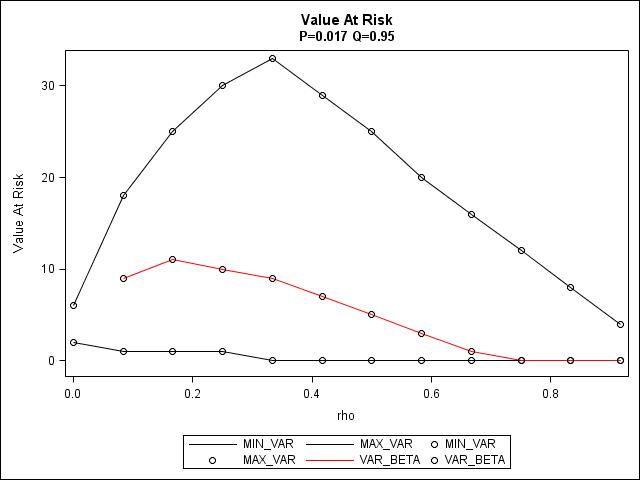}
\end{subfigure}
\begin{subfigure}[b]{0.3\textwidth}
\includegraphics[width=\textwidth]{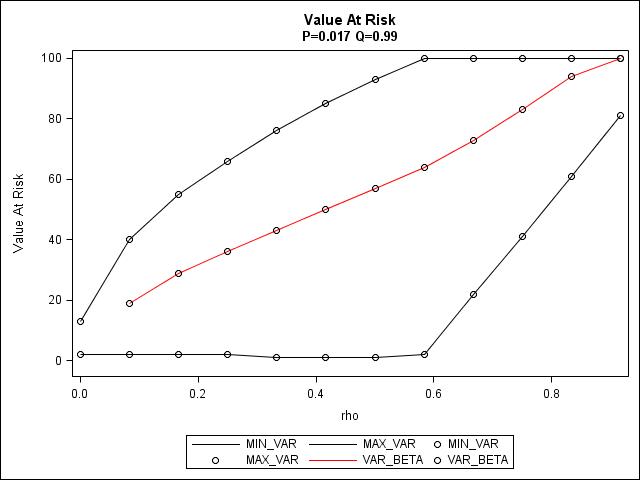}
\end{subfigure}
\end{figure}

\begin{figure}[h!]\caption{VAR  bounds  for $p=26.6\%$ and different $\rho$ and $\beta$-mixing model VAR}\label{VAR3}
\centering
\begin{subfigure}[b]{0.3\textwidth}
\includegraphics[width=\textwidth]{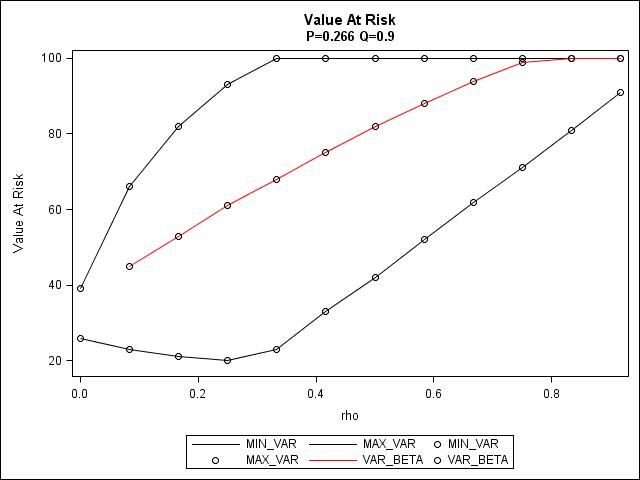}
\end{subfigure}
\begin{subfigure}[b]{0.3\textwidth}
\includegraphics[width=\textwidth]{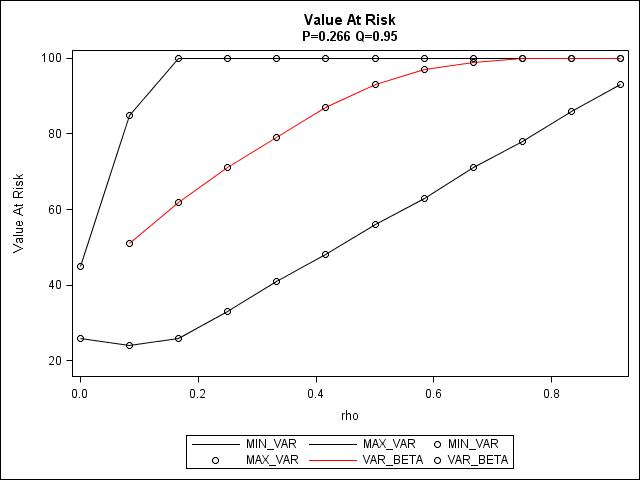}
\end{subfigure}
\begin{subfigure}[b]{0.3\textwidth}
\includegraphics[width=\textwidth]{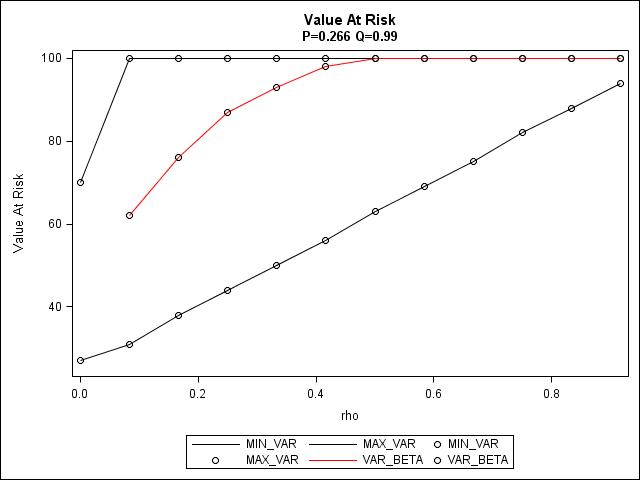}
\end{subfigure}
\end{figure}
\end{document}